\documentclass[sigconf]{acmart}

\copyrightyear{2021}
\acmYear{2021}
\setcopyright{rightsretained}
\acmConference[ISSAC '22]{International Symposium on Symbolic and Algebraic Computation}{July 4--7, 2022}{Lille, France}
\acmBooktitle{International Symposium on Symbolic and Algebraic Computation (ISSAC '22), July 4--7, 2022, Lille, France}
\acmPrice{15.00}
\acmDOI{XX.XXX/XXXXXX.XXXXXX}
\acmISBN{XXXXXXXXXXXXXXXXXXX}

\usepackage[utf8]{inputenc}
\usepackage[T1]{fontenc}

\usepackage{amsmath,  mathrsfs}
\usepackage[all]{xy}
\usepackage{stmaryrd}

\usepackage{tikz}
\usepackage{hyperref}

\usepackage{algorithm}
\usepackage[noend]{algorithmic}

\usepackage{multirow}

\usepackage[english]{babel}
\usepackage{amsthm}
\theoremstyle{plain}
\newtheorem{lem}{Lemma}[section]
\newtheorem{prop}[lem]{Proposition}
\newtheorem{thm}[lem]{Theorem}
\newtheorem{cor}[lem]{Corollary}

\theoremstyle{definition}
\newtheorem{defn}[lem]{Definition}

\theoremstyle{remark}

\newtheorem{rmk}[lem]{Remark}

\newcommand{\N}{\mathbb{N}}
\newcommand{\NN}{\mathbb N}
\newcommand{\Z}{\mathbb{Z}}

\newcommand{\Fp}{\mathbb{F}_p}
\newcommand{\Q}{\mathbb{Q}}
\newcommand{\QQ}{\mathbb{Q}}
\newcommand{\Qp}{\Q_p}

\newcommand{\val}{\mathrm{val}}
\newcommand{\Ke}{K^{\textrm{\rm ex}}}

\DeclareMathOperator{\LT}{\mathrm{LT}}
\DeclareMathOperator{\LC}{\mathrm{LC}}
\DeclareMathOperator{\LM}{\mathrm{LM}}
\DeclareMathOperator{\LTE}{\mathrm{LTE}}
\DeclareMathOperator{\Ecart}{\textup{\textrm{Écart}}}
\DeclareMathOperator{\card}{\mathrm{card}}
\DeclareMathOperator{\Supp}{\mathrm{Supp}}
\DeclareMathOperator{\spoly}{\mathrm{S-Poly}}

\newcommand{\WNF}{\mathsf{WNF}}
\newcommand{\Terms}{\mathrm{Terms}}
\newcommand{\X}{\mathbf{X}}

\renewcommand{\i}{\mathbf{i}}

\renewcommand{\r}{\mathbf{r}}
\newcommand{\s}{\mathbf{s}}

\newcommand{\ifnonempty}[3]{%
  \def\tempa{}%
  \def\tempb{#1}%
  \ifx\tempa\tempb % Empty case
  #3 
  \else            % Non-empty case
  #2
  \fi}

\newcommand{\Kz}{K^\circ}
\newcommand{\KzX}[1][]{K\{ \X \ifnonempty{#1}{; #1}{} \}^\circ}
\newcommand{\KX}[1][]{K \{ \X \ifnonempty{#1}{; #1}{} \}}

\newcommand{\TzX}[1][]{T\{ \X \ifnonempty{#1}{; #1}{} \}^\circ}
\newcommand{\TX}[1][]{T \{ \X \ifnonempty{#1}{; #1}{} \}}
\newcommand{\TT}{\mathbb T}
\newcommand{\TTzX}[1][]{\TT\{ \X \ifnonempty{#1}{; #1}{} \}^\circ}
\newcommand{\TTX}[1][]{\TT \{ \X \ifnonempty{#1}{; #1}{} \}}

\def\todo#1{\ \!\!{\color{red} #1}}

\begin{document}

\fancyhead{}

\title{On Polynomial Ideals And Overconvergence In Tate Algebras}

\author{Xavier Caruso}
\affiliation{Université de Bordeaux,
  \institution{CNRS, INRIA}
  \city{Bordeaux}
  \country{France}}
\email{xavier.caruso@normalesup.org}

\author{
  Tristan Vaccon}
\affiliation{Universit\'e de Limoges;
  \institution{CNRS, XLIM UMR 7252}
  \city{Limoges}
  \country{France}  
  \postcode{87060}  
}
\email{tristan.vaccon@unilim.fr}

\author{
  Thibaut Verron}
\affiliation{Johannes Kepler University, 
  \institution{Institute for Algebra}
  \city{Linz}
  \country{Austria}  
}
\email{thibaut.verron@jku.at}

\thanks{This work was supported by the ANR project CLap--CLap
(ANR-18-CE40-0026-01).
T.~Verron was supported by the Austrian FWF grant P34872.}

\begin{abstract}
In this paper, we study ideals spanned by polynomials or overconvergent series in a Tate algebra. 
With state-of-the-art algorithms for computing Tate Gröbner bases, even if the input is polynomials, the size of the output grows with the required precision, both in terms of the size of the coefficients and the size of the support of the series.

We prove that ideals which are spanned by polynomials admit a Tate Gröbner basis made of polynomials, and we propose an algorithm, leveraging Mora's weak normal form algorithm, for computing it.
As a result, the size of the output of this algorithm grows linearly with the precision.

Following the same ideas,
 we propose an algorithm
 which computes an overconvergent basis for an ideal spanned by overconvergent series.

 Finally, we prove the existence of a universal analytic Gröbner basis for polynomial ideals in Tate algebras, compatible with all convergence radii.
% As an application, we show how working with Tate algebras allows to recover local information about $p$-adic varieties defined by polynomials in $\mathbb{Q}[X_1,\dots,X_n]$.
\end{abstract}



% \vspace{-1mm}
% \ccsdesc[500]{Computing methodologies~Algebraic algorithms}
% \printccsdesc

\vspace{-1.5mm}
\terms{Algorithms, Theory}


\maketitle

%\clubpenalty=10000
\widowpenalty = 10000
\addtolength{\textfloatsep}{-0.45cm} % Distance between float (e.g. [t]) and text

\section{Introduction}
\label{sec:introduction}

The study of $p$-adic geometric objects has taken significant importance in the 20th century, as a crucial component of algebraic number theory.
Beyond polynomials and algebraic geometry, Tate developed a theory of $p$-adic analytic varieties, called rigid geometry.
This theory is now central to many developments in number theory.
The fundamental underlying algebraic object is Tate algebras, that is, algebras of convergent multivariate power series over a complete discrete valuation field $K$ (for instance $\QQ_{p}$).

In earlier papers, the authors examined those Tate series from a computational point of view, with the hope to develop an algorithmic toolbox on par with what is available for polynomials.
The main result was that it is possible to define and compute Gröbner bases of Tate ideals, in a way compatible with the usual theory on polynomials over the residue field (e.g. $\Fp$).
We also examined how different algorithms from the polynomial case transfer to Tate settings.

A key property of Tate series is their convergence radius.
Tate algebras are parameterized by the convergence radius of their series.
If a series is convergent on a certain disk, it is certainly convergent on all disks with smaller radius.
This gives a natural embedding of one Tate algebra into another if the convergence radii of the latter are smaller than those of the former.
This property is a key feature of rigid geometry.
In fact, the canonical embedding of $K[\X]$ into Tate algebras is a particular case of such an embedding, with polynomials seen as series with infinite convergence radius.

Beyond this theoretical interest, recognizing and exploiting such overconvergence properties would help making the algorithms more efficient.
Indeed, a limiting factor of the current algorithms is the cost of the reductions, in particular as the precision grows.
Series with a larger convergence radius are series which converge faster, and thus require to compute fewer terms while reducing.
The challenge in taking advantage of those properties lies in designing algorithms ensuring that this overconvergence property is preserved in the course of the algorithm.

In~\cite{CVV3}, we showed how to generalize the FGLM algorithm~\cite{FGLM93} to Tate algebras.
A result was that this algorithm allows, for zero-dimensional Tate ideals embedded into a Tate algebra with a less restrictive convergence condition, to convert the Gröbner basis.
This opens the possibility, for zero-dimensional ideals, of computing a Gröbner basis in the smaller Tate algebra, where all series have the stronger convergence property, and then using FGLM to convert the result.

In this work, we consider ideals spanned by polynomials in a Tate algebra, from this point of view of overconvergence.
We show that in this case, the ideal admits a Gröbner basis comprised only of polynomials, and we propose an algorithm computing such a basis, and working only with polynomials.
The key ingredient is to use a variant of Mora's weak normal form~\cite{Mora} to compute the head reductions instead of standard reduction.
This algorithm computes reductions up to an invertible factor, with the additional property that all series appearing in the computations are actually polynomials.
In order to do so, it uses specific metrics, called \emph{écarts}, to select the polynomials to use for reduction at each step.
This notion of écart is crucial in proving that the Gröbner basis computation terminates.
In the polynomial case, the écart is simply defined as the difference between the degree of the polynomial and that of its leading term, and in the Tate case we need to refine that with a comparison on the set-difference of the supports of the polynomials.

The resulting algorithm offers a better control for the complexity as a function of the precision.
Concretely, given a set of polynomials in $\QQ[\mathbf{X}]$, and a prime number $p$, we consider the ideal $I$ spanned by the polynomials in $\QQ_{p}\{\mathbf{X}\}$.
Using existing algorithms for Tate Gröbner bases, we can compute a Gröbner basis of $I$ modulo $p^{N}$ for all $N$.
But the output of such a computation will be truncated series, and if we increase the precision $N$, the size of their support typically increases, and even quantifying that growth is not an easy task.
By contrast, the algorithm which we present here only computes polynomials.
So once the precision is large enough, the supports will be completely determined and will not grow anymore.
Asymptotically, the size of the output grows linearly with the precision, and the complexity of the algorithm grows at the same rate as the cost of coefficient arithmetic.

The same idea can be used for ideals spanned by overconvergent series.
With a further refinement of the écart in order to take the valuation of the coefficients into account, we prove that Mora's algorithm allows to compute overconvergent remainders as a result of reducing overconvergent series, and that the modified version of Buchberger's algorithm converges, and computes a basis comprised of overconvergent series.

In later sections, we examine an application of those results for ideals spanned by polynomials in a Tate algebra, namely eliminating variables.
This operation is fundamental in effective algebraic geometry, by allowing to compute various ideal operations such as saturation or intersection.
However, in the Tate setting, due to the nature of the term ordering, computing an elimination ideal by a Gröbner basis may fail.
Concretely, even with an elimination ordering, the leading term of a Tate series is determined by first looking at the valuations of the coefficients, and so it is not enough to look at the leading term to determine whether the series involves the elimination variable or not.

We prove that for ideals spanned by polynomials, using Buchberger's algorithm with Mora reductions, this problem does not appear, and we are indeed able to eliminate variables.
This allows to recover all the usual ideal operations, and in particular proves that polynomial ideals are stable under intersection and radical.
% As an application, we show how it allows to compute local information regarding the integral points of a $p$-adic algebraic variety defined by polynomials over $\QQ$.
% This operation was out of reach of existing tools such as Gröbner bases over $\QQ_{p}$ or Tate Gröbner bases without Mora's algorithm.

Finally, we consider the theory of universal Gröbner bases, that is, sets which are a Gröbner basis for all monomial orderings.
This theory has proved useful in the classical setting, for instance leading to algorithms for change of ordering.
The key result is that any ideal in a polynomial algebra has a finite universal Gröbner basis.
This allows to see the set of Gröbner bases of the ideal as a polyhedral cone, and algorithms wandering on this cone have been developed (see \cite{BM88, Mora-Robbiano, Fukuda-Jensen, GrobnerWalk}).
Furthermore, connections with tropical geometry have been explored (see \cite{gfan,
 Computing_Tropical_Varieties}).

This latter aspect motivates the quest for a similar notion in the Tate setting.
It could pave the way for the computation of the tropical analytic
variety defined by a polynomial ideal (see \cite{Rabinoff}).
However, it is not clear whether all Tate ideals admit a finite universal Gröbner basis.
The last result of this work is a proof that polynomial ideals admit a finite universal analytic Gröbner basis, valid regardless of the choice of the convergence radii.

\section{Setting}

\subsection{Term orders, Tate algebras and ideals}
\label{subsec:Tate_alg}

In order to fix notations, we briefly recall the definition of Tate 
algebras and their theory of Gröbner bases (GB for short).
Let $K$ be a field with valuation $\val$
and let $\Kz$ be the subring of $K$ consisting of
elements of nonnegative valuation. Let $\pi$ be a uniformizer 
of $K$, that is an element of valuation $1$.
Let $\Ke \subset K$, be an exact field.
Typical examples of such a setting are $p$-adic fields like
 $K=\QQ_p$ with
$\Kz=\Z_p,$ $\pi = p$ and $\Ke = \QQ$
or Laurent series fields like $K=\QQ ((T))$ with
$\Kz=\QQ [[T]],$ $\pi = T$ and $\Ke = \QQ(T).$

Let $\r = (r_{1},\dots,r_{n}) \in \QQ^{n}$.
The \emph{Tate algebra} $\KX[\r]$ is defined by:
  \begin{equation}
    \label{eq:1}
    \KX[\r] := \left\{ \sum_{\mathbf{i} \in \NN^{n}} a_{\mathbf{i}}\X^{\mathbf{i}} 
     \text{ s.t. }
     % \begin{cases}
     a_{\mathbf{i}}\in K \text{ and } %\\
     \val(a_\i) - \r{\cdot}\i \xrightarrow[|\i| \rightarrow +\infty]{} +\infty
     % \end{cases}
    \right\}
  \end{equation}
The tuple $\mathbf{r}$ is called the convergence log-radii of the 
Tate algebra.
We define the Gauss valuation of a term 
$a_{\mathbf{i}}\X^{\mathbf{i}}$ as 
$\val_\r(a_{\mathbf{i}}\X^{\mathbf{i}}) = \val(a_\i) - \r{\cdot}\i$, and 
the Gauss valuation of $\sum a_{\mathbf{i}}\mathbf{X}^{\mathbf{i}} 
\in \KX[\r]$ as the minimum of the Gauss valuations of its terms.
The integral Tate algebra ring $\KzX[\r]$ is the subring of $\KX[\r]$ 
consisting of elements with nonnegative valuation.

We fix once and for all a classical \emph{monomial order} $\leq_m$ on the set of
monomials $\X^{\mathbf{i}}$. 
Given two terms $a \X^\i$ and $b \X^{\mathbf{j}}$  (with $a,b \in 
K^\times$), we write $a \X^\i <_\r b \X^{\mathbf{j}}$ if
$\val_\r(a\X^\i) > \val_\r(b\X^{\mathbf{j}})$, or
$\val(a\X^\i) = \val(b\X^{\mathbf{j}})$ 
and $\X^{\mathbf{i}} <_m \X^{\mathbf{j}}$.
The leading term of a Tate series $f=\sum 
a_{\mathbf{i}}\mathbf{X}^{\mathbf{i}} \in \KX[\r]$ is, by
definition, its maximal term, and is denoted by
$\LT_\r (f).$ Its coefficient and its monomial
are denoted $\LC_\r(f)$ and $\LM_\r(f)$, with $\LT_\r (f)=\LC_\r(f)\times \LM_\r(f).$
For $f,g \in \KX[\r]$, we define their S-polynomial as
\[\spoly(f,g) = \frac{\LT_{\r}(g)}{\gcd(LT(f),\LT_{\r}(g))}f - 
\frac{\LT_{\r}(f)}{\gcd(\LT_{\r}(f),\LT_{\r}(g))}g. \]

A Gröbner basis (or GB for short) of an ideal $I$ of $\KX[\r]$ is,
by definition, a family $(g_1, \ldots, g_s)$ of elements of $I$
with the property that for all $f \in I$, there exists an index
$i \in \{1, \ldots, s\}$ such that $\LT_\r(g_i)$ divides $\LT_\r(f)$.
%A Gröbner basis $(g_1, \ldots, g_s)$ is \emph{reduced} if all $\LT(g_i)$'s are monic, minimally generate $\LT(I)$ and, given a term $t$
The following theorem is proved in~\cite{CVV}.

\begin{thm}
\label{theo:GB}
Any ideal of $\KX[\r]$ or $  \KzX[\r]$ admits a Gröbner basis.
\end{thm}

We define the monoid of terms $\TX[\r]$ as the multiplicative monoid 
consisting of the elements $a \X^\alpha$ with $a \in K^\times$ and
$\alpha \in \NN^n$.
We let also $\TzX[\r]$ be the submonoid of $\TX[\r]$ consisting of terms $a 
\X^\mathbf{i}$ for which $\val_\r(a \X^\mathbf{i}) \geq 0$.
The multiplicative group $K^\times$ (resp. $(\Kz)^\times$) embeds 
into $\TX[\r]$ (resp. $\TzX[\r]$). We set:
$$\TTX[\r] = \TX[\r] / K^\times
\quad \text{and} \quad
\TTzX[\r] = \TzX[\r] / (\Kz)^\times.$$

We remark that $G$ is a GB of an ideal $I$ in $\KX[\r]$ 
(resp. $  \KzX[\r]$) if and only if
$\LT_\r(G)$ generates $\LT_\r(I)$ in $\TTX[\r]$ (resp. $\TTzX[\r]$).

\subsection{Polynomial and overconvergent ideals}
\label{sec:polyn-overc-ideals}

The main object of our studies is polynomials and overconvergent series, and the ideals they span.

\begin{defn}
  An ideal of $\KX[\r]$ is called a \emph{polynomial ideal} if it is spanned by polynomials in $K[\X]$.
  
  Let $\mathbf{s} \leq \r$ with respect to component-wise comparison: $\forall i \in \llbracket 1,n\rrbracket, s_i \leq r_i$.
  A series $f = \sum_{i} a_{\mathbf{i}}\mathbf{X}^{\mathbf{i}} \in \KX[\r]$ is called \emph{$\mathbf{s}$-convergent} (or simply \emph{overconvergent} if $\s$ is clear by the context) if
  \begin{equation}
    \label{eq:2}
    \val(a_{\mathbf{i}}) - \mathbf{s} \cdot \mathbf{i}
     \rightarrow_{ \vert i \vert  \to \infty} +\infty.
  \end{equation}
  Equivalently, it means that $f$ is the image of an element of $\KX[\mathbf{s}]$ under the canonical embedding.
  An ideal of $\KX[\r]$ is called \emph{$\mathbf{s}$-convergent}  if it is spanned by $\mathbf{s}$-convergent series.
\end{defn}

\begin{rmk}
  A polynomial ideal in $\KX[\r]$ contains more polynomials than the ideal taken in $K[\X]$.
  For example, let $\r=(0)$ and consider $f=X+pX^{2}$ in $\QQ_{p}[X]$.
  In $\QQ_{p}[X]$, the ideal spanned by $f$ is $\langle f \rangle = \langle X+pX^{2} \rangle$.

  On the other hand, in $\QQ_{p}\{X\}$, $1+pX$ is invertible with inverse $\sum_{i} (-p)^{i}X^{i}$, and the ideal contains $X = \left( \sum_{i} (-p)^{i}X^{i} \right)f$.
\end{rmk}

The following structural results are immediate.
\begin{prop}
  Let $\mathbf{s} \leq \r$.
  Let $I$ and $J$ be two polynomial (resp. $\mathbf{s}$-convergent) ideals in $\KX[\mathbf{r}]$.
  Then:
  \begin{enumerate}
    \item the sum $I+J$ is a polynomial (resp. $\mathbf{s}$-convergent) ideal;
    \item the product $IJ$ is a polynomial (resp. $\mathbf{s}$-convergent) ideal.
  \end{enumerate}
\end{prop}

On the other hand, closure under elimination is not obvious, and therefore closure under intersection or saturation is not immediate.
For intersection, it can be proved using that,
as a completion of a Noetherian
ring, 
$\KX[\r]$ is flat over $\KX[\s]$.
%They can be proved by a flatness argument: the completion of a Noetherian ring, here $\KX[\r]$ over $\KX[\r]$, is always flat.\todo{Thibaut: reference?}
Using Gröbner bases,
we present in Section~\ref{sec:appl-ideal-oper}
a constructive proof for polynomial ideals.

%\textcolor{red}{Thibaut: quid du radical?}

% Finally, we mention that polynomial and overconvergent ideals are not stable under radicals, as the following example shows.
% \begin{expl}
%   Let $K = \CC((T))$, and consider $f = X^{2}+TX^{3}$ spanning the ideal $I$ in $\KX$.
  
% \end{expl}

\section{Polynomial ideals: tools and motivations}

Using elimination, 
we motivate our results
with the closure
of polynomial ideals
under some ideal operations
including intersection 
and saturation.

\subsection{Elimination of one variable}

Let $A = K\{x_1,\dots,x_n;r_1,\dots,r_n\}$
be a Tate algebra with tie-breaking monomial ordering $\leq_{m_A}.$
Let $B = K\{t,x_1,\dots,x_n;r_0,r_1,\dots,r_n\}$ be a Tate algebra
above $A.$
Let $I$ be an ideal of $B.$
We would like to compute a GB of the ideal $I \cap A$
in $A$ (for the  monomial ordering $\leq_{m_A}$).

\begin{prop}
If $r_0 = +\infty$ and $\leq_{m_A}$ is a block-monomial ordering
with $t$ bigger than any monomial not involving $t,$
if $G_B$ is a GB of $I$ for the term ordering defined by the $r_i$'s
and $\leq_{m_A}$, then $G_B \cap A$ is a GB of $I \cap A.$
\end{prop}
\begin{proof}
Firstly, $\left\langle G_B \cap A \right\rangle_A \subset I \cap A.$ 
Now, let us remark that 
if $g \in B$ is such that $\LT(g) \in A$ then $g \in A.$
Indeed,
as $r_0 = + \infty$ and $\leq_{m_A}$ is a block-monomial ordering,
any term $c x^\alpha$ involving $t$ is such that $t>\LT(g)$
so $g$ does not have any term involving $t$.  

As a consequence of this fact,
if $f \in I \cap A$ is divided by $G_B$
then only elements of $G_B$ in $A$, with $\LT$'s in $A$
will be involved, and as $G_B$ is a GB of $I$,
$f$ is reduced to $0.$
Consequently, the same division to $0$ happens for the division of $
f$ by $G_B \cap A$, so $I \cap A \subset \left\langle G_B \cap A \right\rangle_A ,$ which concludes the proof.
\end{proof}

\begin{cor}
For $r_0$ big enough, the previous result for $G_B \cap A$
is preserved, allowing to fit into the framework of algorithms developped
in \cite{CVV, CVV2, CVV3} and in this article.
\end{cor}

\subsection{Application to ideal operations}
\label{sec:appl-ideal-oper}
Following $\S 4$ of \cite{Cox15},
if $I$ and $J$ are 
ideals 
then GBs of $I \cap J$, $I:J$ and $I:J^\infty$
can be computed using
elimination (\textit{e.g.} 
$I \cap J = \left( tI+(1-t)J \right) \cap \KX$).

One motivation for our work
is Corollary \ref{corr:finite_GB_made_of_pol},
stating that any
polynomial ideal in $\KX[\r]$
admits a GB made of polynomials.
It implies the following
stability result
on polynomial ideals: if $I$ and $J$ are
ideals in $\KX[\r]$ generated
by polynomials,
then so are: $I+J$, $IJ$, $I \cap J$, $I:J$ and $I:J^\infty.$
%
%
%\begin{lem}
%Let $I=\left\langle f_1,\dots,f_s \right\rangle$
%and $J=\left\langle g_1,\dots,g_l \right\rangle$
%be two ideals in $\KX[\r].$
%Then the following equalities
%provide algorithm to compute
%Gröbner bases of the
%ideals defined by
%the following operations:
%\begin{align}
%I+J &=\left\langle f_1,\dots,f_s,g_1,\dots,g_l \right\rangle, \\
%IJ &= \left\langle \{f_i g_j, \: i \in \llbracket 1,s\rrbracket,\: j \in \llbracket 1,l\rrbracket \right\rangle, \\
%I \cap J &= \left( tI+(1-t)J \right) \cap \KX, \\
% I:g_1^\infty&=\left(I+\left\langle1-tg_1\right\rangle \right)\cap \KX,\\
% I:J^\infty &= \bigcap_{i=1}^l \left(I:g_i^\infty\right).
%\end{align}
%Furthermore, if $I \cap \left\langle g_1 \right\rangle=\left\langle h_1,\dots,h_p\right\rangle$
%then, \begin{align}I:g_1&=\left\langle h_1/g_1,\dots,h_p/g_1\right\rangle,\end{align}
%and we can then use
%\begin{align}
%I:J&=\bigcap_{i=1}^l \left(I:g_i\right).
%\end{align}
%\end{lem}
%
%We then get the following
%interesting corollary
%on ideals generated
%by polynomials inside
%$\KX[\r]$.

%\begin{cor}
%If $I$ and $J$ are
%ideals in $\KX[\r]$ generated
%by polynomials,
%then so are: $I+J$, $IJ$, $I \cap J$, $I:J$ and $I:J^\infty.$
%\end{cor}

\subsection{Homogenization and dehomogenization}

We will rely on (de)-homogenization at some point in
the computations.
We consign here notations and basic properties.

\begin{defn}
Let $(\cdot)^*$ and $(\cdot)_*$ be the homogenization
and dehomogenization applications between
$K[\X]$ and $K[\X,t].$
If $I$ is an ideal in $K[\X],$
we define $I^*$ to be the homogenization
of this ideal in $K[\X,t].$
\end{defn}

Given $\r \in \QQ^n$, we extend the term order
$<_\r$ to $K[\X,t]$ and $K \left\lbrace \X,t ; \r,0 \right\rbrace$

\begin{defn}
Given two terms $a \X^\alpha t^u$
and $b X^\beta t^v$, we write that
$a \X^\alpha t^u <_{\r,0} b X^\beta t^v$
if:
\begin{itemize}
\item $\val_\r(a\X^\i) > \val_\r(b\X^{\mathbf{j}})$ (which is the same as $\val_{\r,0}(a \X^\alpha t^u) > \val_{\r,0}(b X^\beta t^v)$).
\item $\val_\r(a\X^\i) = \val_\r(b\X^{\mathbf{j}})$ 
and $\deg (\X^\alpha t^u) < \deg (X^\beta t^v)$.
\item $\val_\r(a\X^\i) = \val_\r(b\X^{\mathbf{j}}),$ $\deg (\X^\alpha t^u) = \deg (X^\beta t^v)$
and $\X^{\alpha} <_m \X^{\beta}.$
\end{itemize}
This defines a term order on $K \left\lbrace \X,t ; \r,0 \right\rbrace$. \label{defn:term_order_on_homogenization}
\end{defn}

This order is defined such that dehomogenization preserves leading
terms of homogeneous polynomials of $K[\X,t].$

\begin{lem}
Let $\r \in \Q^n.$
Let $h \in K[\X,t]$ be a homogeneous polynomial.
Then $\LT_{(\r,0)}(h)_*=\LT_\r(h_*).$
Let $f \in K[\X],$
then $\LT_\r(f)=(\LT_{(\r,0)}(f^*))_*.$
\label{lem:dehomogenization_and_LT}
\end{lem}
\begin{proof}
Thanks to the way we defined the term order on
$K \left\lbrace \X,t ; \r,0 \right\rbrace$
in Definition \ref{defn:term_order_on_homogenization},
if $c_\alpha x^\alpha t^u$ and $c_\beta x^\beta t^v$
are two terms of the same total degree such that
$c_\alpha x^\alpha t^u >_{\r,0} c_\beta x^\beta t^v,$
then $c_\alpha x^\alpha  >_{\r} c_\beta x^\beta.$
This is enough for the first part.
For the second part, we can conclude
using $h=f^*$ and the fact that $f=(f^*)_*.$ 
\end{proof}

\section{Weak Normal Forms}

\subsection{Definitions}
\label{sec:definition}

We present here how to adapt Mora's tangent cone algorithm
to compute Weak Normal Forms over Tate algebras.
The main consequence of this notion is that it will allow us, if the generating Tate series are polynomials,
to do all computations on polynomials, avoiding any infinite division.

In this section, we fix
some $\r \in \QQ^n$.
First, we recall the definition of weak normal forms, adapted to the framework of polynomial ideals in Tate algebras.
\begin{defn}
  A \emph{weak normal form} is a map $\WNF: K[\X] \times \mathcal{P}(K[\X]) \to K[\X]$, such that, for all $f \in K[\X]$ and all $G \subseteq K[\X]$, the following holds:
  \begin{enumerate}
    \item $\WNF(0,G) = 0$
    \item If $\WNF(f,G) \neq 0$, then $\LT(\WNF(f,G))$ does not lie in the ideal spanned by the leading terms of $G$
    \item If $f \neq 0$, then there exists $u \in K[\X]$ invertible in $\KX[\r]$ such that $uf - \WNF(f,G)=\sum_{g \in G} u_g g$
    with the $u_g$'s polynomials, $\LT_\r(u_g g) \leq \LT_\r(f)$
    with equality attained at most once.\footnote{This is sometimes called a \textit{strong Gröbner representation} of $f$ by $G.$} 
  \end{enumerate}
\end{defn}

In particular, if $\WNF(f,G)=0$, then $f$ lies in the ideal spanned by $G$.
And if $G$ is a Gröbner basis, it is an equivalence.

\subsection{\'Ecarts}

The first step in order to devise a new version of 
Mora's tangent cone algorithm is to provide a suitable
écart function on polynomials. This function
then drives the division algorithm.
To do so, we adapt the écart functions from \cite{Mora} and \cite{CM} 
to fit into the Tate algebra framework (see
also \cite{SingularIntro} for a general background on standard bases computations).

\begin{defn}
For $f$ a polynomial, we define:
\[\Ecart_1 (f) :=\deg( f) - \deg(\LM_\r(f)).\]
\end{defn}

\begin{defn}
For $h = \sum_u b_u x^u$
and $g= \sum_u c_u x^u$ two polynomials, we define:
\[\Ecart_2(h,g):=\card(\{u : b_u = 0, c_u \neq 0 \}).\]
\end{defn}

%\begin{lem}
%If $f^*$ is the homogenization of $f$.
%Then 
%\begin{align*}
%\LM(f^*)&=t^{\Ecart_1(f)} \LM(f), \\
%\Ecart_1(f) &= \deg_t(\LM(f^*)).
%\end{align*} \label{lem:ecart_homogenisation}
%\end{lem}

\subsection{Mora's Weak Normal Form algorithm}

We first present a simple version of 
Mora's algorithm to compute a Weak Normal Form (WNF) of a polynomial modulo a finite set of polynomials.
It differs from the multivariate division algorithm
by adding intermediate reduced polynomials 
to the list of divisors, which
induces a division which happens, not on the original
divided polynomial, but on one of his multiples by an invertible
polynomial (which does not modify the $\LT_\r$'s).

\begin{algorithm}[H]
	\caption{$\WNF(f,g)$, Mora's Weak Normal Form algorithm}
	\label{algo: Mora's WNF}
	\begin{algorithmic}[1]
		\REQUIRE $f,g_1,\dots,g_s \in K[\X]$
		\ENSURE $h \in K[\X]$ such that for some $\mu,u_1,\dots,u_s \in K[\X],$  $\mu f=\sum u_i g_i +r,$ \\
		$\mu$ is a polynomial such that $\val_\r (\mu- 1)>0,$
					and when $h \neq 0,$ $\LT_\r(h)$ is divisible by no $\LT_\r(g_i)$'s. \\
		Moreover, $\LT_\r (u_i g_i) \leq \LT_\r(f).$	
						
		\STATE	$h:=f$ ; 
		\STATE  $T:=(g_1,\dots,g_s)$ ;
		\WHILE{ $h \neq 0$ and $T_h := \{g \in T, \LT_\r(g) \mid \LT_\r(h) \} \neq \emptyset$}	
		\STATE choose $g \in T_h$ minimizing first 
		$\Ecart_1(g)$ then
		$\Ecart_2\left(h,\frac{\LM_\r(h)}{\LM_\r(g)}g\right)$ ;
		\IF{$\Ecart_1(g) >\Ecart_1(h),$
		or  $\Ecart_2(h,\frac{\LM_\r(h)}{\LM_\r(g)}g)>0$}
		\STATE $T:=T \cup \{h\}$;
		\ENDIF
		\STATE $h:=\spoly(h,g)$ ;
		\ENDWHILE
		\RETURN $h$ ;
	\end{algorithmic}
\end{algorithm}

We may remark that if $\mu \in K[\X]$
is such that $\val_\r(\mu-1)>0$, then
$\mu$ is invertible in $\KX[\r].$

\subsection{Termination}

%Terminaison : page 53 de Greuel-Pfister
%et aussi page 472 de Chan-Maclagan

\begin{lem}
Algorithm \ref{algo: Mora's WNF} terminates.
\label{lem:WNFalgo_terminates}
\end{lem}
\begin{proof}
Let us define
the extended leading
terms (with respect
to $\Ecart_1$) as:
$\LTE \: : \:  K[\X] \rightarrow K[\X,t]$
with $LTE(f)= \LT_\r(f) \times t^{\Ecart_1(f)}.$

Let us assume the 
algorithm does not terminate. It means that 
$T_h$ is never empty.
From Prop 2.8 of \cite{CVV}, there exists some $N$ such that 
$\LTE(T^{(v)})$ is stable for $v \geq N.$

For $v \geq N,$
when $h_v$ is processed,
two possibilities
can occur.
If it is not added 
to $T$ on Line $6$,
it means that 
for the selected reducer $g,$
$\Ecart_1(g) \leq \Ecart_1(h_v).$
If it is added, then
$\LTE(h_v)$ is in 
$\LTE(T^{(v)})$
and hence, there is some
$g \in T^{(v)}$
such that $\LTE(g) \mid \LTE(h_v).$
It means that $\Ecart_1(g) \leq \Ecart_1(h_v)$
and $\LT_\r(g) \mid \LT_\r(h_v).$

Thus, in both cases,
the $g$ selected in Line $4$
has to be such that $\Ecart_1(g) \leq \Ecart_1(h_v).$
In consequence, starting from $v \geq N,$
$\deg(h_v)$
can not increase,
and is upper-bounded by $d.$

Thereafter, the 
amount of 
$\LM_\r(h_j)$'s and
$x^\alpha \LM_\r(g)$'s for $g \in T$ and $\deg (x^\alpha \LM_\r(g))\leq d$
is finite.
Moreover, for 
the polynomials reaching
such an $\LM_\r$,
only a finite amount
of supports are possible.

Therefore, there is some
$N_2>N$ such that
after the $N_2$-th term,
$T$ will not gain
any new support for its
polynomials nor their monomial
multiples of degree $\leq d.$
Then, for $v \geq N_2$,
the minimal $\Ecart_2$
is $0.$ Indeed,
if it is not $0$, then
$h_v$ is added to $T.$
But as $v \geq N_2$,
there is a $g$ with $\LT_\r(g) \mid \LT_\r(h_v)$
and $Support(\frac{\LM_\r(h)}{\LM_\r(g)}g)=Support(h_v),$
and thus $\Ecart_2(\frac{\LM_\r(h)}{\LM_\r(g)}g,h_v)=0,$
which is a contradiction.

Hence, for $m \geq N_2$ necessarily, it means that $\Supp(h_{m+1}) \subsetneq \Supp(h_{m})$
(the leading term of $h_m$ being canceled).
Since the size of the support cannot decrease indefinitely,
the algorithm must terminate.
\end{proof}

\subsection{Correctness}

In order to prove correctness, we extend the algorithm so that the production
of the cofactors is explicit (Algorithm~\ref{algo: Mora's WNF with cofactors}).

\begin{algorithm}
	\caption{Mora's Weak Normal Form algorithm with cofactors}
	\label{algo: Mora's WNF with cofactors}
	\begin{algorithmic}[1]
		\REQUIRE $f,g_1,\dots,g_s \in K[\X]$
		\ENSURE $\mu,u_1,\dots,u_s,h \in K[\X]$ such that  $\mu f=\sum u_i g_i +h$ \\
		when $h \neq 0,$ $\LT_\r(h)$ is divisible by no $\LT_\r(g_i)$'s and $\mu$ is a polynomial such that $\val_\r(\mu-1)>0.$

		\STATE	$h_0:=f,\mu_0=1, u_{1,0}=\dots=u_{s,0}=0,j=0$ ;
		\STATE  $T:=(g_1,\dots,g_s)$ ;
		\WHILE{ $h_j \neq 0$ and $T_{h_j} := \{g \in T, \LT_\r(g) \mid \LT_\r(h_j) \} \neq \emptyset$}	
		\STATE choose $g \in T_{h_j}$ minimizing first 
		$\Ecart_1(g)$ then
		$\Ecart_2\left(h_j,\frac{\LM_\r(h_j)}{\LM_\r(g)}g\right)$ ;
		\IF{$\Ecart_1(g) >\Ecart_1(h),$
		or $\Ecart_2(h_j,\frac{\LM_\r(h)}{\LM_\r(g)}g)>0$}
		\STATE $T:=T \cup \{h_j\}$;
		\ENDIF
		\STATE $x^v := \LM_\r(h_j)/\LM_\r(g),$ $c_v := LC_\r(h_j)/LC_\r(g)$ ;
		\STATE $h_{j+1}:=\spoly(h_j,g)$ \textit{i.e.} $h_{j+1}:=h_j-c_v x^v g $ ;
		\IF{$g=g_m$ for some $1 \leq m \leq s$}
		%\STATE $h_{j+1}:=\spoly(h_j,g)$ \textit{i.e.} $h_{j+1}:=h-c_v x^v g $ ;
		\STATE $u_{m,j+1}:=u_{m,j}+c_vx^v, u_{i,j+1}=u_{i,j}$ for $i \neq m$, $ \mu_{j+1}:=\mu_j$ ;
		\ELSE
		\STATE $g$ was added to $T$ at some previous iteration of the
		algorithm, so $g=h_m$ for some $m<j$ ;
		\STATE $\mu_{j+1}:=\mu_j-c_vx^v\mu_m,$ and for all $i\leq s$, $ u_{i,j+1}:=u_{i,j}-c_vx^v u_{i,m}$; 
		\ENDIF

%		\ELSE
%		\STATE $r_{j+1}:= r_j+\LT(h_j),$ $h_{j+1}:=h_{j}-\LT(h_j)$, $u_{i,j+1}=u_{i,j},$ $\mu_{j+1}:=\mu_j$ ;
		\STATE $j:=j+1$ ;
		\ENDWHILE
		\RETURN $\mu_{j}, u_{1,j},\dots,u_{s,j},h_j$ ;
	\end{algorithmic}
\end{algorithm}

Correctness then comes from the following loop invariant:

\begin{lem}
  \label{lem:correctness}
For any $j \geq 0$,
\begin{enumerate}
\item $\mu_j f=h_j+\sum_i u_{i,j} g_i,$
\item $\val_\r (\mu_j -1)>0,$ 
\item $\LT_\r(u_{i,j} g_i) \leq \LT_\r(f),$ with equality attained at most once, and if so, always with the same $i$ for all $j$;
\item $\LT_\r(h_{j+1}) < \LT_\r(h_{j}).$ 
\end{enumerate}
\end{lem}
\begin{proof}
It is clearly true when entering the first loop.

The equality for the third item is attained once
after the end of the first loop.

Inside the loop, there is no difficulty when
the reduction
is performed by one of the initial $g_i$'s.
One applies the fourth item to ensure
that no second $\LT_\r(u_{i,j} g_i)$ reaches $\LT_\r(f).$

When the divisor $g$ was added to $T$ at a previous iteration
of the algorithm, \textit{i.e.} $g=h_m$ for some $m<j,$
then the situation is the following.
Firstly, the preservation of the fourth item is clear.

Then, as $m<j,$ we get from the fourth item of the loop invariant that 
$\LT_\r(h_j)<\LT\r(g)$ and also
$\LT_\r(c_vx^v g)=\LT_\r(h_j).$
It implies that
$\val_\r(c_vx^v)>0$.
Hence, as $\val_\r(\mu_j -1)>0,$
the same is true for $\mu_{j+1}:=\mu_j-c_vx^v \mu_m$ 
and the second item is preserved.

From $\mu_j f=h_j+\sum_i u_{i,j} g_i ,$
and $\mu_m f=h_m+\sum_i u_{i,m} g_i m,$
one gets 
$(\mu_j-c_vx^v \mu_m) f=(h_j-c_vx^v h_m)+\sum_i (u_{i,j}-c_vx^v u_{i,m}) g_i $
so
$\mu_{j+1} f=h_{j+1}+\sum_i u_{i,j+1} g_i +r_{j+1},$
and the first item is preserved.

As $\val_\r(c_v x^v) >0,$
then $\LT_\r(c_v x^v u_{i,m})<\LT_\r(u_{i,m}),$
so 
\begin{align*}
\LT_\r(u_{i,j}-c_v x^v u_{i,m}) &\leq \max (\LT_\r(u_{i,j}), \LT_\r(c_v x^v u_{i,m})), \\
 &\leq \max (\LT_\r(u_{i,j}), \LT_\r(u_{i,m})), 
\end{align*}
which is enough to obtain that the third item is preserved and concludes
the proof.
\end{proof}

\begin{cor}
  Algorithm \ref{algo: Mora's WNF} computes a weak normal form.
\end{cor}
\begin{proof}
  We verify the three items of the definition of weak normal forms.
  If $f = 0$, the algorithm immediately returns $0$.

  Assume that $\WNF(f,G) \neq 0$.
  This implies that after the last loop of the algorithm, $T_{h} = \emptyset$, and since $\LT_\r(T)$ contains the leading terms of $G$, $\LT_\r(\WNF(f,G))$ is not divisible by any of the $\LT_\r(G)$.

  Finally, the third item follows from the third item of Lemma~\ref{lem:correctness}.
\end{proof}

\begin{cor}
If $G$, a finite set of polynomials, is a GB of $I_\r \subset \KX[\r]$,
then for any polynomial $f \in I_\r$,
$\WNF(f,G)=0$. \label{cor:WNF_by_a_GB}
\end{cor}
\begin{proof}
If $G$ is a GB of $I_\r$, then when dividing $f \in I_\r$,
on Line 3, $T_h$ is never empty.
Indeed, from the first item of the loop invariant,
$h_j=\mu_j f-\sum_i u_{i,j} g_i$ means that $h_j \in I_\r.$
Consequently, the algorithm can only terminate
if $h$ reaches $0.$
\end{proof}

\section{Buchberger's algorithm with WNF}
\label{sec:buchberger}
\subsection{Description of the algorithm}
\label{sec:descr-algor}

We prove Buchberger's criterion following
the lines of \S 3.2 of \cite{CVV}.
We rely on a small variation of the
technical Lemma 3.6 of \cite{CVV},
which is a generalization of~\cite[Sec.~2.10, Prop.~5]{Bu65}:

\begin{lem}
\label{lem:cancel_and_Spol}
Let $h_1,\dots,h_m \in \KX[\r] $ and $t_1,\dots,t_m \in \TX[\r]$.
We assume that the $\LT_\r(t_i h_i)$'s all have the same image in
$\TTX[\r]$ and that
$\LT_\r(\sum_{i=1}^m t_i h_i) < \LT_\r(t_1 h_1)$. Then 
\[\sum_{i=1}^m t_i h_i =
\sum_{i=1}^{m-1} t'_i {\cdot}\spoly(h_i,h_{i+1})
+ t'_m {\cdot}h_m  \]
for some $t'_1, \ldots, t'_m \in \TX[\r]$ such that
$\LT_\r(t'_i {\cdot}\spoly(h_i,h_{i+1}))< \LT_\r(t_1 h_1)$
for $i \in \{1, \ldots, m{-}1\}$
and $\LT_\r(t'_m {\cdot}h_m)< \LT_\r(t_1 h_1).$
\end{lem}

\begin{proof}
By assumption, there exist $\alpha \in \NN^n$ and 
$d_1, \ldots, d_m \in K^\times$ such that $\LT_\r(t_i h_i) = 
d_i \X^{\alpha}$ for all~$i \in \{1,\ldots,m\}$. Moreover all the 
$d_i$'s have the same valuation, say $\mu$.
Then $\val(\sum_i d_i) > \mu$.
We define $p_i=\frac{t_i h_i}{d_i}$, so that $t_i h_i = d_i p_i$.
Then
\begin{align*}
\sum_i t_i h_i 
&= d_1 (p_1-p_2)+ (d_1 + d_2) (p_2-p_3)+ \cdots + \\[-2mm]
&\phantom{={}} (d_1+\dots+ d_{m-1})(p_{m-1}-p_m)+ (d_1+\dots+ d_m) p_m.
\end{align*}
Observing that $p_i - p_{i+1} = t'_i 
\cdot \spoly(h_i, h_{i+1})$ with $t'_i$ a term, we get the first $t_i'$'s
and $t_m'=\frac{\sum_i d_i}{d_m}t_m.$

Then, since $\val(\sum_i d_i) > \mu,$ clearly $\LT_\r(t'_m {\cdot}h_m)< \LT_\r(t_1 h_1).$
Moreover, for any $i \in \{1, \ldots, m{-}1\}$, 
due to the cancellation in $\spoly(h_i,h_{i+1}),$
$\LT_\r (p_{i}-p_{i+1}) < \LT_\r (p_{i})=\frac{\LT_\r (t_{i}h_i)}{d_i}.$
Considering valuations,
$\sum_{k\leq i} d_k \cdot \frac{\LT (t_{i}h_i)}{d_i} \leq \LT_\r (t_{i}h_i).$
Consequently, $\LT_\r(t'_i {\cdot}\spoly(h_i,h_{i+1}))< \LT_\r(t_i h_i),$
which concludes the proof.
\end{proof}

\begin{prop}[Buchberger's criterion]
$G$ is a GB of $I_\r$ if and only if $G$ generates $I_\r$ and $\WNF(\spoly(g_i,g_j),G)=0$ for all pairs $g_i,g_j \in G$.
\end{prop}
\begin{proof}
The $\Rightarrow$ part is direct thanks to Corollary \ref{cor:WNF_by_a_GB}.

Let us prove the $\Leftarrow $ part.
Let $f \in I$ be such that $\LT_\r (f) \notin \left\langle \LT_\r (G) \right\rangle.$
As $G$ generates $I$, $f$ can be written as $f = \sum_{i=1}^s h_i g_i$
for some Tate series $h_i$'s in $\KX[\r].$

Let $t = \max_i \LT_\r (h_i g_i).$
As $\LT_\r (f) \notin \left\langle \LT_\r (G) \right\rangle$, then
$\LT_\r (f) <t.$
Consequently, among the decompositions of $f$ using $G$, there is one such that
$t$ is minimal.

Let $J$ be the set of indices $i$ such that  $\LT_\r (h_i g_i)= c_i t$
for some $c_i \in  O_K^\times.$
Let $t_i = \LT_\r(h_i)$ for $i \in J.$
Let $h = \sum_{i \in J} t_i g_i.$ We have $\LT_\r(h) <t$
and $\card (J) \geq 2$
as a cancellation has to appear.
We apply Lemma \ref{lem:cancel_and_Spol}: there exist terms $t_{j,l}'$
and $t'$ and an index $j_0 \in J$ such that
$h = \sum_{j,l \in J} t_{j,l}' \spoly(g_j,g_l)+t' g_{j_0}$
and $\LT_\r (t' g_{j_0}) < t$
and $\LT_\r (t_{j,l}'\spoly(g_j,g_l)) < t.$
We can compute the WNF of the polynomial $\spoly(g_j,g_k)$ by $G$
and we get some invertible polynomial $u_{j,l}$ and polynomials $v_i^{(j,l)}$ such that: $ u_{j,l}\spoly(g_j,g_l) = \sum_{i=1}^s v_i^{(j,l)} g_i$ with
$\LT_\r (v_i^{(j,l)} g_i) \leq \LT_\r (\spoly(g_j,g_l)).$

Multiplying by $u_{j,l}^{-1}$
and summing those decompositions, we get that
$\sum_{j,l \in J} t_{j,l}' \spoly(g_j,g_l)= \sum_{i=1}^s w_i g_i$
with
$\LT_\r(w_i g_i) \leq \max_{j,l} \LT_\r (u_{j,l}^{-1} v_i^{(j,l)} g_i)=\max_{j,l} \LT_\r ( v_i^{(j,l)} g_i)$.
So $\LT_\r(w_i g_i)$ is less than or equal to $\max_{j,l} \LT_\r (t_{j,l}'\spoly(g_j,g_l))$ and strictly smaller than $t$.
Summing all summands we then obtain a new decomposition of $f$  contradicting the minimality of $t.$
\end{proof}

\begin{algorithm}[H]
	\caption{Buchberger's algorithm with Mora's WNF}
	\label{algo: Buchberger with WNF}
	\begin{algorithmic}[1]
		\REQUIRE $F:=(f_1,\dots,f_s)$ a list of polynomials in $K[\X].$
		\ENSURE $G$ a list of polynomials in $K[\X]$ which is a GB of $\left\langle F \right\rangle$

		\STATE	$G:=F$ ; 
		\STATE  $P:=\{ (f,g) \mid f,g \in G, \: f \neq g \}$ ;
		\WHILE{ $P \neq \emptyset$ }	
			\STATE choose and remove $(f,g)$ from $P$ ;
			\STATE $h := \WNF(\spoly(f,g), G)$;
			\IF{ $h \neq 0$}
				\STATE $P:= P \cup \{ (h,f) \mid f \in G \}$ ;
				\STATE $G:=G \cup \{h \}$ ;
			\ENDIF		
		\ENDWHILE
		\RETURN $G$ ;
	\end{algorithmic}
\end{algorithm}

\begin{prop}
Algorithm \ref{algo: Buchberger with WNF} terminates and is correct.
\end{prop}
\begin{proof}
Correctness comes from Buchberger's criterion.
Termination is a consequence of Prop 2.8 of \cite{CVV}. 
\end{proof}

\begin{cor}
If $I$ is generated by polynomials, then
Algorithm \ref{algo: Buchberger with WNF}
provides a GB of $I_\r \subset  \KX[\r]$ made of polynomials of $I.$
\label{corr:finite_GB_made_of_pol}
\end{cor}
\begin{proof}
If $f \in I$ and $G \subset I$ are polynomials, then
$\WNF(f,G)$ is a polynomial of $I.$
As the $\spoly$ considered in Algo. \ref{algo: Buchberger with WNF}
are polynomials in $I$, we obtain the result.
\end{proof}

\subsection{Precision and effective computations}
\label{subsec:precision_and_Kex}
We may remark firstly that,
as we wrote all properties and proofs in terms
of $\LT$'s, the algorithms of $\S 3$ and $4$
are valid over $\KzX[\r].$
In particular, if we work with
$\r = (0,\dots,0),$ 
then no division in $K$ is involved:
as in \cite{CVV}, working at finite precision,
no loss of absolute precision can occur.

Secondly, if we work in $\Ke \subset K$, all computations take place in $\Ke$. Hence if $(f_1,\dots,f_s) \in \Ke [\X]$ and $\r \in \QQ^n$,
Algorithms \ref{algo: Mora's WNF} and \ref{algo: Buchberger with WNF} provide an algorithm working over $\Ke$ to compute a GB of $I_\r$ made of polynomials
in $\Ke [\X],$ without having to deal with any precision issue.

\subsection{Toy Implementation}

A toy implementation of the algorithms
of this Section is available
here: \url{https://gist.github.com/TristanVaccon}.
We present some timings
and features of the Algorithm
in Appendix \ref{sec:appendix_timings}
on page \pageref{sec:appendix_timings}.

\section{Mora's WNF and overconvergence}

We now consider the case of overconvergent series, and present a version of Mora's weak normal form algorithm for that case.

\subsection{\'Ecarts for overconvergence}

Let $f,g \in \KX[\s], \s \in \Q^n, \r \in \Q^n, \s \geq \r$.
We define écarts adapted to computation
over $\KX[\r]$ for series belonging also to
$\KX[\s].$

\begin{defn}
We define the $\s$-support of $f=\sum_{\alpha \in \N^n} c_\alpha \X^\alpha \in \KX[\s]$ as:
\[\Supp_\s(f) = \left\lbrace \alpha \textrm{ s.t. } \val_\s (c_\alpha \X^\alpha) = \val_\s (f) \right\rbrace.\]
Since $f \in \KX[\s],$
$\Supp_\s(f)$ is finite. 
Then, we define the $(\s,\r)$-degree of $f$ as:
\[\deg_{\s,\r}(f) = \max_{\alpha \in \Supp_\s(f)} (\s-\r)\cdot \alpha. \]
\end{defn}

\begin{defn}
We define:
\begin{align*}
\Ecart_{\s,\r,0}(f)&:=\val_\s(\LT_\r(f))-\val_\s(f), \\
\Ecart_{\s,\r,1}(f)&:=\deg_{\s,\r}(f)-\deg_{\s,\r}(\LT_\r(f))
\end{align*}
\end{defn}

\begin{lem}
For $f \in \KX[\s],$
$i \in \left\lbrace 0,1 \right\rbrace, 
\Ecart_{\s,\r,i}(f)\geq 0.$
\end{lem}
\begin{proof}
For $\Ecart_{\s,\r,0}(f)$, it is a direct consequence
of the definition of $\val_\s.$

Now, let us take some $\alpha \in \Supp_\s(f)$
such that $(\s-\r)\cdot \alpha = \deg_{\s,\r}(f).$
Let $c_\alpha$ be the coefficient of $\X^\alpha$ in $f.$
Let $\LT_\r(f)=c_\beta \X^\beta.$
Then, by definition,
$\val_\s(c_\beta \X^\beta) \geq \val_\s(f)=\val_\s(c_\alpha \X^\alpha),$
and $\val_\r(c_\alpha \X^\alpha) \geq \val_\r(c_\beta \X^\beta).$
Thus,
\begin{align*}
\val (c_\beta)-\s \cdot \beta &\geq \val (c_\alpha)-\s \cdot \alpha, \\
\val (c_\alpha)-\r \cdot \alpha &\geq \val (c_\beta)-\r \cdot \beta, \text{ and then } \\
\s \cdot (\alpha-\beta) \geq \val(c_\alpha)&-\val(c_\beta) \geq \r \cdot(\alpha-\beta),
\end{align*}
which implies that $(\s-\r)\cdot (\alpha-\beta)\geq 0.$
Since $\deg_{\s,\r}(f)=(\s-\r)\cdot \alpha$
and $\deg_{\s,\r}(\LT_\r(f))=(\s-\r)\cdot \beta,$
we can conclude that $\Ecart_{\s,\r,1}(f) \geq 0.$
\end{proof}

\subsection{WNF algorithm for overconvergent series}

The algorithm is straightforward, using the adapted notions of écarts.

\begin{algorithm}[H]
	\caption{$\textsf{WNF}(f,g, s,r)$, Mora's overconvergent Weak Normal Form algorithm}
	\label{algo: Mora's WNF for overconvergence}
	\begin{algorithmic}[1]
		\REQUIRE $f,g_1,\dots,g_s \in \KX[\s], \s \in \Q^n, \r \in \Q^n, \s \geq \r$
		\ENSURE $h \in \KX[\s]$ such that for some $\mu,u_1,\dots,u_s \in \KX[\s],$  $\mu f=\sum u_i g_i +h,$ \\
					when $h \neq 0,$ $\LT_\r(h)$ is divisible by no $\LT_\r(g_i)$'s and $\mu$ is invertible in $\KX[\r]$.
		Moreover, $\LT_\r (u_i g_i) \leq \LT_\r(f).$	
						
		\STATE	$h:=f$ ; 
		\STATE  $T:=(g_1,\dots,g_s)$ ;
		\WHILE{ $h \neq 0$ and $T_h := \{g \in T, \LT_\r(g) \mid \LT_\r(h) \} \neq \emptyset$}	
		\STATE choose $g \in T_h$ minimizing first 
		$\Ecart_{\s,\r,0}(g)$, then $\Ecart_{\s,\r,1}(g)$ ;
		\IF{$\Ecart_{\s,\r,0}(g) >\Ecart_{\s,\r,0}(h),$
		or  $\Ecart_{\s,\r,1}(g)>\Ecart_{\s,\r,1}(h)$}
		\STATE $T:=T \cup \{h\}$;
		\ENDIF
		\STATE $h:=\spoly(h,g)$ ;
		\ENDWHILE
		\RETURN $h$ ;
	\end{algorithmic}
\end{algorithm}

\subsection{Correctness and convergence}

\begin{lem}
If $g \in T_{h_m}$ is such that:
\begin{itemize}
\item $\Ecart_{\s,\r,0}(g) \leq \Ecart_{\s,\r,0}(h_m)$,
\item $\Ecart_{\s,\r,1}(g) \leq \Ecart_{\s,\r,1}(h_m)$,
\end{itemize}
and if $t=\LT_\r(h_m)/\LT_\r(g)$
and $h_{m+1}=h_m-tg,$ then
\[\val_s(h_{m+1}) \geq \val_s(h_m). \]
In case of equality, then moreover,
\[ \deg_{\s,\r}(h_{m+1})\leq \deg_{\s,\r}(h_m).\]\label{lem:ecarts_vals_degre_s_r}
\end{lem}
\begin{proof}
Since $\Ecart_{\s,\r,0}(g) \leq \Ecart_{\s,\r,0}(h_m)$
and $\Ecart_{\s,\r,0}(g)=\Ecart_{\s,\r,0}(tg)$,
then $\val_\s(\LT_\r(tg))-\val_\s(tg) \leq \val_\s(\LT_\r(h_m))-\val_\s(h_m)$.
Moreover, $\LT_\r(tg)=\LT_\r(h_m)$, so
$\val_\s(tg) \geq \val_\s(h_m).$
By the ultrametric inequality,
we then obtain that $\val_\s(h_{m+1}) \geq \val_\s(h_m).$

Now, if $\val_\s(h_{m+1}) \geq \val_\s(h_m),$
we prove that $\deg_{\s,\r}(h_{m+1})\leq \deg_{\s,\r}(h_m).$
Since $\Ecart_{\s,\r,1}(g)=\Ecart_{\s,\r,1}(tg)$,
then the second hypothesis means that
$\deg_{\s,\r}(tg)-\deg_{\s,\r}(\LT_\r(tg)) \leq \deg_{\s,\r}(h_m)-\deg_{\s,\r}(\LT_\r(h_m)).$
From the equality $\LT_\r(tg)=\LT_\r(h_m)$, it follows that
$\deg_{\s,\r}(tg) \leq \deg_{\s,\r}(h_m)$.
As $h_{m+1}=h_m-tg,$ then
$\deg_{\s,\r}(h_{m+1}) \leq \max \big( \deg_{\s,\r}(h_m), \deg_{\s,\r}(tg) \big),$ and we can conclude. 
\end{proof}

\begin{prop}
If $\r < \s$ then
either Algorithm \ref{algo: Mora's WNF for overconvergence}
terminates in a finite number of steps,
or both $\LT_\r(h_m)$ and $\LT_\s(h_m)$ converge
to $0.$ \label{prop:terminaison_convergence_WNF_surconvergent}
\end{prop}
\begin{proof}
Let us assume that
Algorithm \ref{algo: Mora's WNF for overconvergence}
does not terminate for
some inputs $\s > \r$ 
and 
$f,g_1,\dots,g_s \in \KX[\s].$

As we do eliminate
successively the $\LT_\r(h_m)$'s,
then by design, $\LT_\r(h_m)$ converges to zero.
 
Let $d_1,d_2 \in \N$
be such that for
any $f \in \KX[\s],$
$d_1 \val_\s(f) \in \Z$
and $d_2 \deg_{\s,\r}(f) \in \Z.$

Let us define the
extended leading term
of $h \in \KX[\s]$
as: $\LTE(h):=U^{d_1 \Ecart_{\s,\r,0}(h)}V^{d_2 \Ecart_{\s,\r,1}(h)}\LT(h) \in K[\X,U,V].$

Then, there is some
$N_1 \in \N$ such that 
for $m \geq N_1,$
the monomial ideal of 
$K[\X,U,V]$ generated by
the $\LTE$'s of the series
in $T$ is constant
(thanks to Prop. 2.8 of \cite{CVV}).
Thus for $m \geq N_1,$
if $h_m$
is not added to $T$
at the end of the $\textbf{while}$ loop,
then there is some $g \in T_{h_m}$
such that $\Ecart_{\s,\r,0}(g)\leq \Ecart_{\s,\r,0}(h_m).$
If it is added, then
by definition of $N_1$,
it means that there is
some $g \in T$
such that $\LTE(g)$
divides $\LTE(h_m)$,
and this implies
that $\LT(g) \mid \LT(h_m)$
and $\Ecart_{\s,\r,0}(g) \leq \Ecart_{\s,\r,0}(h_m).$

So in both cases, 
$\Ecart_{\s,\r,0}(g) \leq \Ecart_{\s,\r,0}(h_m).$

Then, 
if $h_m$ is not added to
$T$, it means that
the minimal $g$
satisfies
$\Ecart_{\s,\r,1}(g) \leq \Ecart_{\s,\r,1}(h_m).$
If it is added to $T$,
then again, 
by definition of $N_1$,
it means that there is
some $g \in T$
such that $\LTE(g)$
divides $\LTE(h_m)$,
and this implies
that $\LT(g) \mid \LT(h_m)$
and $\Ecart_{\s,\r,0}(g) \leq \Ecart_{\s,\r,0}(h_m)$
and $\Ecart_{\s,\r,1}(g) \leq \Ecart_{\s,\r,1}(h_m).$
So in both cases,
the minimal $g$
for the reduction
satisfies that
$\Ecart_{\s,\r,1}(g) \leq \Ecart_{\s,\r,1}(h_m).$

We can then apply Lemma
\ref{lem:ecarts_vals_degre_s_r}:
for any $m \geq N_1$,
$\val_\s (h_{m+1}) \geq \val_\s (h_m)$
and in case of equality,
$\deg_{\s,\r}(h_{m+1}) \leq \deg_{\s,\r}(h_m).$

Consequently, $\val_\s(h_m)$
is a non-decreasing sequence
in $\frac{1}{d_1}\Z.$
Hence, either it goes
to $+\infty$,
or there is some $N_2 \geq N_1$ such that $\val_\s(h_m)$
is constant for $m\geq N_2.$

Let us assume that 
we are in this second case.
Then $\deg_{\s,\r}(h_m)$
is non-increasing (for $m\geq N_2$) and thus, upper-bounded.
Let $m \geq N_2$ 
and $t=c_\alpha \X^\alpha$
a term of $h_m$ in $\Supp_\s(h_m).$

Then $\val_\s(h_m)=\val_\s(t)$ and
\begin{align*}
 \val_\r (h_m) &\leq \val_\r (t) \leq \val_\s(t)+(\s-\r)\cdot \alpha \\
 &\leq \val_\s(h_m)+\deg_{\s,\r}(h_m).
\end{align*}
Both $\val_\s(h_m)$
and $\deg_{\s,\r}(h_m)$
are upper-bounded, while
$\val_\r (h_m) \rightarrow+\infty.$
This is a contradiction.

Consequently,
$\val_\s (h_m) \rightarrow + \infty,$
which concludes the proof.
\end{proof}

\begin{prop}
Algorithm \ref{algo: Mora's WNF for overconvergence} is
correct and
\textit{mutatis mutandis}, 
computes a weak normal form.
\end{prop}
\begin{proof}
\textit{Mutatis mutandis},
the loop invariant in
Lemma \ref{lem:correctness}
is still valid.
When $f$ does not reduce to
zero by $g_1,\dots,g_s$, there is no difficulty
as Algorithm \ref{algo: Mora's WNF for overconvergence}
terminates in a finite
number of steps, and
$\mu,g_1,\dots,g_s$
are polynomials,
with $\mu$ invertible in
$\KX[\r].$
When $f$ reduces to zero,
we proved in Lemma \ref{lem:ecarts_vals_degre_s_r},
that $\val_s(h_m)$
is eventually increasing
and going to $+\infty.$
We showed in the proof
of Prop. \ref{prop:terminaison_convergence_WNF_surconvergent}
that eventually, $T$
is constant.
It then proves that,
for the $g$ on Line 4,
for $c_\nu x^\nu=
\left(\frac{\LT_\r (h_m)}{\LT_\r(g)} \right),$
then $\val_\s (c_\nu x^\nu) \rightarrow+\infty.$
This is enough to prove
that the $\mu,u_1,\dots,u_s$
such that $\mu f= \sum_i u_i g_i$
are in $\KX[\s]$
as expected.
\end{proof}

\begin{rmk}
Section \ref{sec:buchberger}
can be extended
with (almost) no modification
to compute
GB in $\KX[\s]$ of an $\s$-convergent
ideal of $\KX[\r].$
One just needs to replace $K[\X]$
by $\KX[\s]$ and use
Algo. \ref{algo: Mora's WNF for overconvergence}
in Buchberger's algorithm.
\end{rmk}

\section{Universal Gröbner basis}

In this Section, we prove that 
a polynomial ideal
can only have a finite
number of distinct initial ideals
for varying log-radii $\r.$
To do so, we first prove
the result for homogeneous
ideals by adapting the classical
proof for polynomial ideals and then use
homogenization
to generalize the result to
non-homogeneous ideals.

\subsection{Homogeneous ideal}

The classical proof that a polynomial
ideal has only finitely many initial
ideals from page 427 of \cite{Cox05} (see also \cite{Sturmfels})
can be adapted to our setting
by relying on the following Lemma.

\begin{lem}
If $I \subset K[\X]$ is a homogeneous ideal,
if $\r \in \Q^n,$
and if $F=(f_1,\dots,f_s) \in I^s$ are homogeneous
polynomials which do not form 
a GB of $I_\r$, then there exists
some homogeneous polynomial $g \in I$
such that no term of $g$ is
divisible by any of the $\LT_\r(f_i)$'s. \label{lem:reduction_homogeneous_pol}
\end{lem}
\begin{proof}
Since $F$ is not a GB of $I_\r$,
there exists some term $c x^\alpha \in \LT_\r(I_\r)$
such that $c x^\alpha \notin \left\langle \LT_\r(f_1), \dots, \LT_\r(f_s) \right\rangle.$
By the density of $I$ in $I_\r$ there is some polynomial
$h \in I$ such that $\LT_\r(h) = c x^\alpha.$
Since $I$ is homogeneous, we can assume that so is $h$.

By performing the tropical row-echelon algorithm
of \cite{Vaccon:2015} (Algorithm 1) on a Macaulay matrix
consisting of $h$ and the multiples
of the elements of $F$ of degree $\deg (h)$,
we obtain $g$ such that no term of $g$ is
divisible by any of the $\LT_\r(f_i)$'s.
\end{proof}

Using linear algebra along
the same lines, 
we get the existence
of polynomial reduced Gröbner
bases.

\begin{lem}
If $I \subset K[\X]$ is a homogeneous ideal,
if $\r \in \Q^n,$
then there exists
$G$ a reduced Gröbner basis
of $I_\r$
made of finitely many homogeneous polynomials
of $I$. \label{lem:rGB_for_homogeneous_polynomaisl}
\end{lem}
\begin{proof}
Thanks to Corollary
\ref{corr:finite_GB_made_of_pol},
we get $H$, a GB
or $I_\r$ made of
polynomials of $I.$
Since $I$ is homogeneous
we can assume that in addition,
they are all homogeneous.
Then again, for any $g \in G$, 
we can perform inter-reduction
by performing the tropical row-echelon algorithm
of \cite{Vaccon:2015} (Algorithm 1) on a Macaulay matrix
consisting of $g$ and the multiples
of the elements of $G \setminus \{g\}$ of degree $\deg (g)$.
This is enough to conclude.
\end{proof}

\begin{prop}
Let $I \subset K[\X]$ be a homogeneous ideal.
Then the set $Terms(I):=\{\LT(I_\r) \: \diagup \: \r \in \Q^n \}$
is finite. \label{prop:finitess_of_LTs_of_homogeneous_ideal}
\end{prop}
\begin{proof}
Suppose that $\Terms(I)$ is infinite.
For any $M \in \Terms(I)$, we write $\leq_M$ for a term
order defined by an $\r$ such that $\LT(I_\r)=M.$
Let $\Sigma := \{ \leq_M \: \diagup \: M \in \Terms(I) \}.$
Our assumption states that $\Sigma$ is infinite.

Let $f_1 \in I$ be a homogeneous polynomial.
Since $f_1$ has finitely many terms,
by the pigeonhole principle, there is an infinite set 
$\Sigma_1 \subset \Sigma$ and a term $m_1$ of 
$f_1$ such that for all $\leq_M \in \Sigma_1,$ $\LT_{\leq_M}(f_1)=m_1.$
Suppose that for some $\leq_1 \in \Sigma_1$
defined by some $\r_1$,
$(f_1)$ is a GB of $I_{\r_1}$.
Then, let $\leq \in \Sigma_1$ be
defined by some $\r$.
We prove that $(f_1)$ is then a GB of $I_\r.$
Indeed, by Lemma \ref{lem:reduction_homogeneous_pol},
if    $(f_1)$ is not a GB of $I_\r$
there is some $h \in I$ such that
no term of $h$ is divisible by $\LT_{\leq}(f_1).$
Since $\LT_{\leq}(f_1)=\LT{\leq_1}(f_1)$
and $(f_1)$ is a GB of $I_{\r_1}$ this is
a contradiction.
Consequently, for any $\leq \in \Sigma_1$ 
defined by some $\r$, $(f_1)$ is a GB of $I_\r$
with $\LT_\leq (f_1)=m_1$.
However, this can not be the case as our assumption
was that there are infinitely many
elements in $\Sigma_1$
all defining distinct $\LT$'s for $I$.
Therefore,  $(f_1)$ is not a GB of $I_\r.$

By Lemma $\ref{lem:reduction_homogeneous_pol}$
there is some homogeneous $f_2 \in I$
such that no term of $f_2$ is divisible by $m_1.$
Then again, since $f_2$ has finitely many terms,
by the pigeonhole principle, there is an infinite set 
$\Sigma_2 \subset \Sigma_1$ and a term $m_2$ of 
$f_2$ such that for all $\leq_M \in \Sigma_2,$ $\LT_{\leq_M}(f_2)=m_2$ (and also since $\Sigma_2 \subset \Sigma_1$,
$\LT_{\leq_M}(f_1)=m_1$).

The same argument as above shows that 
for any $\leq \in \Sigma_2$ 
defined by some $\r$, $(f_1,f_2)$ is not GB of $I_\r.$
Then again, by Lemma $\ref{lem:reduction_homogeneous_pol}$
there is some homogeneous $f_3 \in I$
such that no term of $f_3$ is divisible by any of $(m_1,m_2).$
Since $f_3$ has finitely many terms,
by the pigeonhole principle, there is an infinite set 
$\Sigma_3 \subset \Sigma_2$ and a term $m_3$ of 
$f_3$ such that for all $\leq_M \in \Sigma_3,$ $\LT_{\leq_M}(f_3)=m_3$ (and also since $\Sigma_3 \subset \Sigma_2$,
$\LT_{\leq_M}(f_1)=m_1, \LT_{\leq_M}(f_2)=m_2$).

Continuing the same way, we produce a descending chain of
infinite subsets $\Sigma \supset \Sigma_1 \supset \Sigma_2 \supset \Sigma_3 \supset \dots$
and an infinite strictly ascending chain of ideals
$\left\langle m_1 \right\rangle \subset \left\langle m_1,m_2 \right\rangle \subset \left\langle m_1,m_2,m_3 \right\rangle \subset \dots$
 in $\TTX.$
This contradicts Prop. 2.8 of \cite{CVV} and concludes the proof.
\end{proof}

\begin{thm}
Let $I \subset K[\X]$ be 
a homogeneous ideal.
Then there exists a finite set $G \subset I\subset K[\X]$
made of homogeneous 
polynomials
which is a \textit{universal analytic Gröbner basis}
of $I$: for any $\r \in \Q^n,$
$G$ is a GB of $I_r.$ \label{thm:universal_analytic_GB_homogeneous_case}
\end{thm}
\begin{proof}
By Prop \ref{prop:finitess_of_LTs_of_homogeneous_ideal},
there are only finitely
many initial ideals possible.
We prove that for two
term-orders $\leq_1$
and $\leq_2$ (defined
by $\r_1$ and $\r_2$),
if they define the same
initial ideal, then
they have the same
reduced Gröbner basis.
Indeed, let $G_1$
and $G_2$ be the 
reduced Gröbner bases
given by Lemma \ref{lem:rGB_for_homogeneous_polynomaisl}.
They have the same $\LT$'s.
Let $g_1 \in G_1$
and $g_2 \in G_2$
having a common $\LT$.
Then $g_1-g_2 \in I$
with no monomial
divisible by any of the
$\LT(G_i)$'s.
Hence $g_1=g_2$, and
$G_1$ and $G_2$ are equal
up to permutation.

Consequently, all 
term orders giving rise
to the same initial
ideal share the same
reduced GB.
Consequently, only
a finite amount of reduced
GB for the $I_\r$'s are
possible.
By concatening all of them,
we obtain the desired
universal analytic Gröbner 
basis.
\end{proof}

\subsection{Non-Homogeneous ideal}

\begin{lem}
Let $I \subset K[\X]$ be a polynomial ideal and 
$\r \in \Q^n.$
Let $(h_1,\dots,h_s)$ be a finite Gröbner basis of $(I^*)_{(r,0)} \subset K \left\lbrace \X,t ; \r,0 \right\rbrace$
made of homogeneous polynomials of $I^*$ (hence in $K[\X,t]$).
Then $(h_{1,*},\dots,h_{s,*})$ is a Gröbner basis of $I_r.$ \label{lem:dehomogenization_of_GB}
\end{lem}
\begin{proof}
Firstly, due to being
dehomogenization of elements of $I^*$,
the $h_{i,*}$'s are in $I$.

Secondly,
by Corollary \ref{corr:finite_GB_made_of_pol}, it is enough
to check that for any $f \in I$,
$\LT_\r(f)$ is divisible by one of the $\LT_\r(h_{i,*})$'s.

Let $f \in I.$ Then $f^* \in I^* \subset (I^*)_{(r,0)}$
so there is some $i$ such that $\LT_{(\r,0)}(h_i)$
divides $\LT_{(\r,0)}(f^*).$
Then thanks to Lemma \ref{lem:dehomogenization_and_LT},
$\LT_\r (f)=\LT_{(\r,0)}(f^*)_*$,
$\LT_\r (h_{i,*})=\LT_{(\r,0)}(h_i)_*$,
and monomial divisibility is preserved by dehomogenization.
So $\LT_\r (h_{i,*})$ divides $\LT_\r (f)$
and the proof is complete.
\end{proof}

We can then prove the main theorem of this section.

\begin{thm}
Let $I \subset K[\X]$ be an ideal.
Then the set $Terms(I):=\{\LT(I_\r) \: \diagup \: \r \in \Q^n \}$
is finite.
\end{thm}
\begin{proof}
Thanks to Lemma \ref{lem:dehomogenization_of_GB}, there
is a surjection from $\Terms(I^*)$
to $\Terms(I).$
The first set is finite thanks to Proposition \ref{prop:finitess_of_LTs_of_homogeneous_ideal},
so the second is also, which concludes the proof.
\end{proof}

We can also obtain the existence of 
universal Gröbner bases for any polynomial
ideal in $K[\X].$
 
\begin{thm}
Let $I \subset K[\X]$ be an ideal.
Then there exists a finite set $G \subset I\subset K[\X]$
which is a \textit{universal analytic Gröbner basis}
of $I$: for any $\r \in \Q^n,$
$G$ is a GB of $I_r.$ \label{cor:universal_GB}
\end{thm}
\begin{proof}
Thanks to Lemma \ref{lem:dehomogenization_of_GB},
it is enough to dehomogenize
a universal analytic GB
of $I^*$
to obtain the desired
universal analytic GB of $I.$
\end{proof}

\subsection{New challenges}

One can relate the previous result
to the Remark 8.8 of \cite{Rabinoff} on the
foundations of computations in tropical analytic geometry,
on universal analytic GB
and on tropical bases.

We say that $F \subset I$ is a tropical
basis of $I$ if
for any $\r \in \QQ^n$:
there is $g \in I$
such that $\val_\r(g)$
is reached by only one term
if and only if there is
$f \in F$ such that
$\val_\r(f)$
is reached by only one term.

It leaves us
with the following
challenges:

\begin{enumerate}
\item Give an algorithm
to compute a universal analytic
Gröbner basis of a polynomial
ideal.
\item Give an algorithm
to compute a tropical basis
of a polynomial ideal.
\item Generalize
universal analytic GB
to overconvergent
ideals
or to varying 
center of polydisks of
convergence. 
\end{enumerate}

We shall remark that in our
context, due to the fact
that we take the valuation
of the coefficients 
into account, then, contrary to
the classical case of
Gröbner fans for polynomials
over a field, 
the Gröbner complex
is in general not a cone.

\begin{table*}
\centering
\begin{tabular}{|l|l|l|c|c|c|c|c|c||c|}
\cline{4-10} 
\multicolumn{3}{c|}{\textbf{Timings (s)}} & \multicolumn{7}{|c|}{\textbf{Entry precision in $\Qp[\X]$ or $\Qp \left\lbrace \X ; (0,\dots,0) \right\rbrace$}}   \\ 
\hline 
\textbf{system} \vphantom{$2^{2^2}$} & $p$& \textbf{algo} &  $2^4$&$2^5$ &$2^6$&$2^7$ &$2^8$&$2^9$&$2^{20}$  \\ 
\hline 
\hline
%\multirow{2}{*}{Cyclic 3} & \multirow{2}{*}{2} &  Mora & ? & ? & ? & ? & ? & ? & ? & ?\\
%\cline{3-11}
%& & Vapote & ? & ? & ? & ? & ? & ? & ?& ?\\
%\hline
\multirow{2}{*}{Cyclic 5} &\multirow{2}{*}{2} &  Mora  & $\infty$ & $\infty$ & $\infty$ & $\infty$ & $\infty$ & $\infty$& $\infty$\\
\cline{3-10}
& & Vapote &  0.86 & 1.0 & 1.5 & 2.3 & 3.8 & 7.2 & $\infty$\\
\hline

\multirow{2}{*}{Katsura 3} & \multirow{2}{*}{2} & Mora   & 0.031 & 0.047 & 0.031 & 0.063 & 0.047 & 0.032 & 0.5\\
\cline{3-10}
& & Vapote &   0.063 & 2.2 & 140 & 4500 & $\infty$ & $\infty$ & $\infty$\\
\hline
\multirow{2}{*}{Katsura 6} &\multirow{2}{*}{2} &  Mora &   1.2 & 0.98 & 0.94 & 1.0 & 1.1 & 1.0 & 2.3\\
\cline{3-10}
& & Vapote &  170 & $\infty$ & $\infty$ & $\infty$ & $\infty$ & $\infty$ & $\infty$\\
\hline

%\multirow{2}{*}{Random} & \multirow{2}{*}{7} & Mora & 0.92 & 1.1 & 0.86 & 0.83 & 0.86 & 0.86 & 0.89 & 270\\
%\cline{3-11}
%& & Vapote & 0 & 0 & 0 & 0 & 0 & 0 & 0& 28\\
%\hline
\end{tabular}
\caption{Precision and timing for Algo \ref{algo: Buchberger with WNF} and Vapote \cite{CVV2}. The log-radii is $(0,\dots,0).$
}
\label{table_des_timings}
\end{table*}

\newpage

\bibliographystyle{plain}

\begin{thebibliography}{}

\end{thebibliography}


\begin{thebibliography}{99}
  \renewcommand{\itemsep}{0em}
  % \renewcommand{\itemsep}{0.14em}

%\bibitem{Ben19}
%Benedetto, R. L.,
%\newblock{ Dynamics in one non-archimedean variable},
%\newblock{ American Mathematical Soc. (Vol. 198).}

  \bibitem{BM88}
  Bayer D., Morrison I.
  \newblock {Standard bases and geometric invariant theoryI. Initial ideals and state polytopes},
  \newblock {Journal of Symbolic Computation,
Volume 6, Issues 2–3,
1988.}

\bibitem{Computing_Tropical_Varieties}
Bogart T., Jensen A. N.,  Speyer D.,  Sturmfels B., Thomas R. R.
\newblock{Computing Tropical Varieties}
\newblock{
J. Symb. Comput. 42 (2007), no. 1-2}

\bibitem{Bu65}
  Buchberger, B.,
\newblock{{Ein Algorithmus zum Auffinden der Basiselemente des Restklassenringes nach einem nulldimensionalen Polynomideal (An Algorithm for Finding the Basis Elements in the Residue Class Ring Modulo a Zero Dimensional Polynomial Ideal)}},
\newblock English translation in J. of Symbolic Computation, Special Issue on Logic, Mathematics, and Computer Science: Interactions. Vol. 41, Number 3-4, Pages 475--511, 2006


  \bibitem{CVV}
  Caruso, X., Vaccon T., Verron T.,
  \newblock {Gröbner bases over Tate algebras},
  \newblock {in Proceedings: ISSAC 2019, Beijing, China.}


  \bibitem{CVV2}
  Caruso, X., Vaccon T., Verron T.,
  \newblock {Signature-based algorithms for Gröbner bases over Tate algebras},
  \newblock {in Proceedings: ISSAC 2020, Kalamata, Greece.}

  \bibitem{CVV3}
  Caruso, X., Vaccon T., Verron T.,
  \newblock {On FGLM Algorithms With Tate Algebras},
  \newblock {in Proceedings: ISSAC 2021, Saint-Petersburg, Russia.} 

  \bibitem{CM}
  Chan A., Maclagan D.,
  \newblock  {{Gröbner bases over fields with valuations}},
  \newblock Math. Comp. 88 (2019), 467-483.

  \bibitem{GrobnerWalk}
   Collart S., Kalkbrenner M., Mall D. ,
  \newblock  {{Converting bases with the Gröbner walk.}},
  \newblock J.
Symbolic Comp. 6 (1997), 209–217.

\bibitem{Cox15}
Cox D., Little John.,  O'Shea D.,
\newblock {\em Ideals, varieties, and algorithms}.
\newblock Undergraduate Texts in Mathematics. Springer, Cham, fourth edition,
2015.


\bibitem{Cox05}
Cox D., Little John.,  O'Shea D.,
\newblock {\em Using Algebraic Geometry}.
\newblock Graduate Texts in Mathematics, Volume 185, Springer Science \& Business Media, 2005.

\bibitem{FGLM93}
  Faugère, J.-C., Gianni, P., Lazard, D., Mora, T.,
  \newblock{Efficient computation of zero-dimensional Gr{\"o}bner bases by change of ordering},
  \newblock{J. of   Symbolic Computation 16~(4), 329--344, 1993}



\bibitem{Fukuda-Jensen}
  Fukuda K.,  Jensen A. N.,
   Thomas R. R. ,
  \newblock{Computing Gröbner fans},
  \newblock{Math. Comp. 76 (2007)} 
  

  \bibitem{SingularIntro}
  Greuel G.-M., Pfister G.,
  \newblock  {{A Singular Introduction to Commutative Algebra}},
  \newblock Springer-Verlag Berlin Heidelberg 2008.

\bibitem{gfan}
     Jensen, A.N.,
     \newblock{{G}fan, a software system for {G}r{\"o}bner fans and tropical varieties},
     \newblock{Available at \url{http://home.imf.au.dk/jensen/software/gfan/gfan.html}}
 

  \bibitem{Mora}
  Mora T.,
  \newblock {An Algorithm to Compute the Equations of Tangent Cones},
  \newblock {in: Proceedings EUROCAM 82, Lecture Notes in Comput. Sci. (1982).} 

  \bibitem{Mora-Robbiano}
  Mora T., Robbiano L.,
  \newblock {The Gröbner fan of an ideal, },
  \newblock {J. Symbolic Comput. 6 (1988), no. 2-3.} 

 \bibitem{Rabinoff}
 Rabinoff J.,
\newblock {Tropical analytic geometry, Newton polygons, and tropical intersections},
\newblock {Advances in Mathematics 229(6), 2010.}

  \bibitem{Sage}
  \newblock {{S}ageMath, the {S}age {M}athematics {S}oftware {S}ystem ({V}ersion 9.2)}, The Sage Development Team, 2020, \url{http://www.sagemath.org}

\bibitem{Sturmfels}
Sturmfels B.,
\newblock {Gröbner Bases and Convex Polytopes,}
\newblock {American Mathematical Society, Univ. Lectures Series, No 8, Providence, Rhode Island, 1996.} 

  \bibitem{Tate}
  Tate J.,
  \newblock  {{Rigid analytic spaces}},
  \newblock  {Inventiones Mathematicae} \textbf{12}, {1971}, {257--289}

\bibitem{Vaccon:2015}
Vaccon T.,
\newblock {Matrix-F5 Algorithms and Tropical Gr{\"{o}}bner Bases Computation},
\newblock {Proceedings of the 40th International Symposium on Symbolic and Algebraic Computation, {ISSAC} 2015, Bath, United Kingdom. Extended version
in the Journal of Symbolic Computation, Dec. 2017}.

\end{thebibliography}

\newpage

\section{Appendix: timings}
\label{sec:appendix_timings}
We present here with Table \ref{table_des_timings} some timings for our toy implementation
of Algorithm \ref{algo: Buchberger with WNF}
acting on special systems in $\QQ[\X] \subset \QQ_2[\X] \subset \QQ_2 \left\lbrace \X ; 0,\dots,0 \right\rbrace.$
The $\infty$ symbols means that 12 hours were not enough for the algorithm to terminate.

Most of the time the algorithms of \cite{CVV,CVV2} vastly
outperforms our implementation (as seen in the Cyclic case).

However, this is not always the case
and with the Katsura systems, 
our implementation displays two remarkable features
of our algorithm:
\begin{itemize}
\item Reductions can be significantly faster: no problem
with reductions converging possibly slowly to zero \footnote{An example of such a reduction slowly converging to $0$ in the
algorithms of \cite{CVV, CVV2} is the reduction of $X$ by $X-pX^2$ for log-radii $0$, leading to intermediate remainders $pX^2,p^2X^3,\dots,p^k X^{k+1},\dots$.}
\item The dependency on the precision can be significantly
smaller than that of the algorithms of \cite{CVV,CVV2},
allowing in some cases many orders of magnitude
of additional digits in less time.
\end{itemize}

%Here, Random is a special example in $\Qp[x,y,z,t]$ with random initial polynomials
%$[f_1,f_2,f_3,f_4]$, sparse, of respective total degrees $[1,3,3,4]$
%and $f_4 \in \left\langle f_1,f_2,f_3 \right\rangle.$

Please note 
the special shape
of the Katsura 6 system
in $\Qp[X_1,\dots,X_6]$
for $\r=(0,\dots,0)$
and $p=2$:
its defining polynomials
already contains 
the leading monomials
$X_1, X_2, X_4$,
explaining in part
why this computation
is not as hard as for
classical Gröbner bases.

One can try all
examples at \url{https://gist.github.com/TristanVaccon}.

\end{document}